\documentclass[11pt,a4paper]{amsart}

\usepackage{subfigure}
\usepackage[utf8]{inputenc}
\usepackage[usenames,dvipsnames]{xcolor}
\usepackage{amsmath, amssymb, amsthm}
\usepackage{graphicx}
\usepackage{enumerate}
\usepackage{url}
\usepackage[colorlinks=true,linkcolor=Blue,citecolor=Green]{hyperref}
\usepackage[capitalise]{cleveref}
\crefname{enumi}{}{}
\usepackage[margin=5pt]{caption}
\usepackage[norelsize,boxed,linesnumbered]{algorithm2e}
\SetAlCapSkip{5pt}
\SetAlCapFnt{\small}
\SetAlCapSty{textnormal}
\SetAlCapNameFnt{\small}

\crefname{algocfline}{Algorithm}{Algorithms}

\newcommand{\R}{\mathbb R}

\newcommand{\N}{\mathbb N}

\newcommand{\Ep}{\mathfrak{E}}

\renewcommand{\emptyset}{\varnothing}

\newtheorem{theorem}{Theorem}[section]
\newtheorem{lemma}[theorem]{Lemma}
\newtheorem{proposition}[theorem]{Proposition}

\newtheorem{corollary}[theorem]{Corollary}
\newtheorem{conjecture}[theorem]{Conjecture}
\theoremstyle{definition}
\newtheorem{definition}[theorem]{Definition}

\DeclareMathOperator{\TT}{\mathcal{T}} 
\DeclareMathOperator{\h}{\mathcal{H}} 
\DeclareMathOperator{\lex}{\mathcal{H}} 
 
\DeclareMathOperator{\lexo}{\mathcal{L}}

\DeclareMathOperator{\EP}{\mathrm{EdgePairs}} 

\newcommand{\scalprod}[2]{\left\langle #1,#2 \right\rangle}

\DeclareMathOperator{\E}{\mathcal{E}}

\begin{document}

\title[Bottleneck Partial-Matching Voronoi Diagrams]{Bottleneck Partial-Matching Voronoi Diagrams and Applications}

\author{Matthias Henze}
\address{Institut f\"ur Informatik, Freie Universit\"at Berlin, Takustrasse 9, 14195 Berlin, Germany}
\email{matthias.henze@fu-berlin.de}

\author{Rafel Jaume}
\address{Institut f\"ur Informatik, Freie Universit\"at Berlin, Takustrasse 9, 14195 Berlin, Germany}
\email{jaume@inf.fu-berlin.de}

\thanks{Research by the first author is supported by the ESF EUROCORES program EuroGIGA-VORONOI, (DFG): RO 2338/5-1, and by the second author by ``Obra Social la Caixa'' and DAAD}

\begin{abstract}
Given two point sets in the plane, we study the minimization of the bottleneck distance between a point set $B$ and an equally-sized subset of a point set $A$ under translations.
We relate this problem to a Voronoi-type diagram and derive polynomial bounds for its complexity that are optimal in the size of~$A$.
We devise efficient algorithms for the construction of such a diagram and its lexicographic variant, which generalize to higher dimensions. 
We use the diagram to find an optimal bottleneck matching under translations, to compute a connecting path of minimum bottleneck cost between two positions of~$B$, and to determine the maximum bottleneck cost in a convex polygon.
\end{abstract}

\maketitle

\section{Introduction}

Applications often demand algorithms to find an occurrence of a point pattern in a given cloud of points. 
Using a suitable cost function, it is common to define a similarity measure between the pattern and the point cloud as the minimum cost among the images of the pattern under a set of allowed transformations. 
One of the most studied similarity measures between finite point sets~$A$ and~$B$ in $\R^d$ is the \emph{directed Hausdorff distance}, which is the maximum of the (Euclidean) distances from each point in~$B$ to its nearest neighbor in~$A$.
For some applications in robotics and pattern recognition, it is required that each point of the smaller set is matched to a distinct point in the bigger one. 
The resulting distance is called the \emph{bottleneck distance} and was introduced for equally-sized sets in~\cite{Alt} as 
\[\Delta(B,A)=\underset{\sigma:B \hookrightarrow A} \min \; \underset{b \in B} \max \; \Vert b-\sigma(b) \Vert ,\]
where $\|\cdot\|$ denotes the Euclidean norm and the minimum is taken over all injections from~$B$ into~$A$.
In contrast to the directed Hausdorff distance, the bottleneck distance has the advantage of being symmetric for equally-sized sets.
On the other hand, it is harder to compute, since the points cannot be regarded independently. 
Note that there might be several matchings that minimize the bottleneck distance, even when all the distances between points are distinct.
When this is to be avoided, considering the matching that lexicographically minimizes the distances between matched points helps to break some ties; cf.~\cite{Burkard,Croce,Sokkalingam}.

In this paper, we are interested in a dynamic version of the bottleneck distance.
More precisely, we want to efficiently compute, among all translated copies of $B$ with respect to~$A$, one attaining the minimum bottleneck distance; that is, $\min_{t \in \R^d}\Delta (B+t,A)$.
This problem will be called \emph{bottleneck partial-matching under translations}.
It was introduced for equally-sized point sets in the plane by Alt, Mehlhorn, Wagener \& Welzl~\cite{Alt}, who gave an algorithm running in $O(n^6 \log n )$ time for point sets of size $n$. 
Their bound was later improved to $O(n^5 \log^2 n)$ by Efrat, Itai \& Katz~\cite{Efrat}.

To the best of our knowledge, bottleneck matching under translations has not been studied with the focus on algorithms whose complexity is sensitive to the size of the smaller set.
In order to do so, we associate Voronoi-type diagrams to the problem, which we call \emph{bottleneck diagrams} and \emph{lex-bottleneck diagrams}, respectively.
This follows an idea of Rote~\cite{Rote2010} who partitioned the space of translations according to the (partial) matching that minimizes the least-squares distance between translated copies of~$B$ and~$A$ (cf.~\cite{Henze,ESA,ESA_journal} for follow-up studies).
Our bottleneck diagrams partition~$\R^d$ into polyhedral cells that correspond to locally-optimal (lexicographic-)bottleneck matchings. 

Our motivation to investigate these diagrams does not restrict to solving the bottleneck partial-matching problem under translations only.
We moreover aim to provide a structure that may be either used for preprocessing or may be adjusted towards other problems that are based on the computation of the bottleneck distance in various translated positions of the point sets.
The applications at the end of the paper exemplify this utility of the bottleneck diagrams.

A non-archival abstract containing parts of our studies appeared in~\cite{ISAAC}.

\subsubsection*{Our Results}

In \cref{sect:diagrams}, we formally introduce the Voronoi-type diagrams before investigating their basic properties and combinatorial complexity.
It turns out that there exists a \mbox{lex-}bottleneck diagram (and, hence, a bottleneck diagram) of complexity $O(n^2k^6)$ for any given planar point sets $A,B\subset\R^2$ with $k=|B|\leq|A|=n$ (see \cref{thm:bot}), and that this bound cannot be improved with respect to the size of~$A$.
For point sets $A,B\subset\R^d$ of higher dimensions we obtain in \cref{cor:firstbound} that there is a lex-bottleneck diagram of complexity $O(n^{2d}k^{2d})$.
Based on this complexity result, we devise algorithms in~\cref{sect:construction} that construct these polyhedral subdivisions of~$\R^d$ and at the same time compute a (lexicographic-)bottleneck matching for each of the cells of the subdivision, which is necessary to solve the bottleneck matching problem under translations.
This is achieved with an overhead of $O(k^2)$ for the bottleneck diagrams, and $O(k^4)$ for the lexicographic variant (see \cref{thm:constr,thm:constrHD}).
The matching problem under translations can then be solved for the bottleneck case in time $O(n^2k^8)$, and for the lexicographic variant in time $O(n^2k^{10})$, if the point sets are planar (see \cref{thm:optmatch}).
In higher dimensions the time bounds are $O(n^{2d}k^{2d+2})$ and $O(n^{2d}k^{2d+4})$, respectively.
Finally, \cref{thm:safpath,thm:coverrad} show how we can use the bottleneck diagrams to efficiently compute a path of minimum bottleneck cost between two positions of a pattern in the plane, or how to determine what we call the cover radius of a polygon.

\subsubsection*{Comparison to previous work}

Although neither Alt et al.~\cite{Alt} nor Efrat et al.~\cite{Efrat} consider the bottleneck matching problem for different-sized point sets, their methods can be adapted to this situation without major difficulties.
In \ref{appendix}, we elaborate on such an analysis of their algorithms and derive the time bounds $O(n^3 k^3 \log n)$ and $O(n^2 k^3 \log^2 n)$, respectively, where $k$ is the size of the smaller set and~$n$ the size of the bigger one.
This shows that the adapted algorithm by Efrat~et~al. outperforms our procedure, that runs in $O(n^2k^8)$, already for fairly small values of $k$, in particular for any $k = \Omega(\log^{2/5} n)$.
Still, the use of bottleneck diagrams is conceptually different from previous methods, and has the advantage of being applicable to solve the translative matching problem in any dimension and, moreover, with respect to the lexicographic bottleneck cost.
No exact algorithms were known for higher dimensions previously, however, there exist approximation algorithms for the bottleneck matching problem (see~\cite{Efrat}).

\subsubsection*{Organization of the paper}

In the next section, we introduce (lexicographic) bottleneck matchings between two finite point sets and investigate corresponding minimization diagrams and their properties.
After these basics, we define our main objects of study, bottleneck partial-matching Voronoi diagrams, and analyze their combinatorial complexity in \cref{sect:diagrams}, before addressing construction algorithms for these structures in \cref{sect:construction}.
Finally, in \cref{sect:applications}, we apply the bottleneck diagrams to solve the bottleneck partial-matching problem and related questions.

\section{Bottleneck and Lexicographic Bottleneck Matchings}
\label{sect:botmatchings}

In this section, we introduce bottleneck matchings and discuss the minimization diagram corresponding to the bottleneck partial-matching problem under translation.
The issues we encounter explain our approach to the definition of the bottleneck diagrams in \cref{sect:diagrams}.

Throughout the paper, we assume that we are given two point sets $A,B \subset \R^d$ with $k=|B| \leq |A|=n$ and that~$B$ is allowed to be translated.
We use the term \emph{edge} for a pair of points $(a,b) \in A \times B$ and denote it by $ab$ for short.
The \emph{length} of the edge $ab$ is defined as the Euclidean distance $\|b-a\|$.
In this context, we identify every injection of $B$ into $A$ with the \emph{matching}, i.e., the set of edges, it induces.
The cost of such a matching varies according to a parameter representing the position of the point set~$B$.

\begin{figure}[t]
\centering
\hfill
\subfigure[Two point sets.]{\includegraphics[scale=.8,page=5]{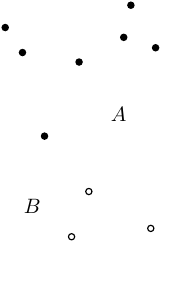} \label{fig:pointset}}
\hfill
\subfigure[The minimization diagram of $\E$.]{\includegraphics[scale=.8,page=2]{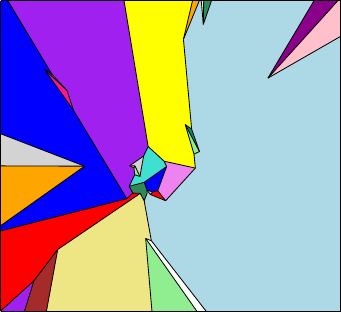}\label{fig:botdiagram}}
\hfill
\subfigure[A coarsening of a bottleneck diagram of $A$ and $B$.]{\includegraphics[scale=.8,page=3]{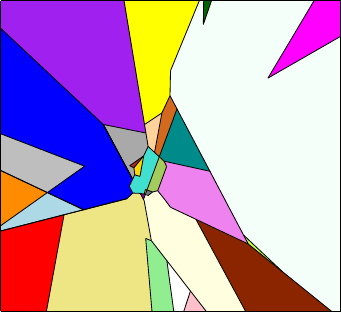}\label{fig:lexdiagram}}
\hfill\,
\caption{The minimization diagram of $\E$ and a coarsening of a bottleneck diagram for a pair of point sets.}
\end{figure}

\begin{definition}\label{def:bot}
Let $A, B \subset \R^d$ be finite point sets with $|B| \leq |A|$.
A \emph{bottleneck matching} for~$A$ and $B$ is a matching that minimizes 
\[ f(\sigma) = \underset{b \in B} \max \; \| b-\sigma(b) \|  \text{ among all the matchings } \sigma:B \hookrightarrow A. \]
The \emph{bottleneck cost} of a matching $\sigma:B \hookrightarrow A$ is the function $f_\sigma:\R^d\to\R_{\geq0}$ defined as
\[ f_\sigma(t)=\underset{b \in B} \max \; \| b+t-\sigma(b) \|^2, \textrm{ for all } t \in \R^d. \]
The \emph{bottleneck value function} $\E:\R^d\to\R_{\geq0}$ is defined by
\[ \E(t)= \underset{\sigma: B \hookrightarrow A} \min \; f_\sigma(t), \textrm{ for all } t \in \R^d. \]
\end{definition}

Note that a bottleneck matching is defined in terms of the Euclidean distance while the functions $f_\sigma$ and $\E$ depend on the square of this value.
This squaring is harmless and will be convenient later on. 

By definition, the function~$\E$ is piecewise quadratic and induces a partition of~$\R^d$ into polyhedral regions which is usually called the \emph{minimization diagram} of~$\E$.
\cref{fig:pointset,fig:botdiagram} show a pair of planar point sets and the corresponding minimization diagram whose regions are colored according to the edges attaining the bottleneck value. 
Note that some regions are not convex and some are even disconnected. 
More precisely, the red and the blue regions consist of two connected components. 
The possible non-convexity of these regions is more concisely illustrated in \cref{Fig:nonconvex}.
The pictured disks certify that the drawn edges are the ones attaining the bottleneck value for the three aligned positions of the small point set. 

\begin{figure}[t]
\centering
 \includegraphics[page=7]{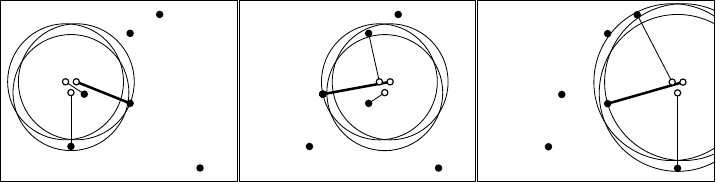}
 \caption{Three positions of a pair of point sets and a bottleneck matching for each of them showing that the minimization diagram of the corresponding $\E$ function has a non-convex region.}
 \label{Fig:nonconvex}
\end{figure}

Moreover, the leftmost and the rightmost matchings in \cref{Fig:nonconvex} are the only optimal matchings for the respective positions of the small point set.
On the other hand, they have the same longest edge, i.e., the quadratic functions in the corresponding regions of the minimization diagram coincide.
Note in addition that, disregarding the longest edge, these two matchings have disjoint sets of edges.
Conversely, it can happen that a matching is the unique bottleneck matching for two open sets of positions contained in different regions of the minimization diagram of~$\E$.
An instance of this situation is illustrated in~\cref{fig:bottleneckEpslon}; remember that different longest edges in the matching correspond to different regions.

\begin{figure}[ht]
\centering
\includegraphics[page=14]{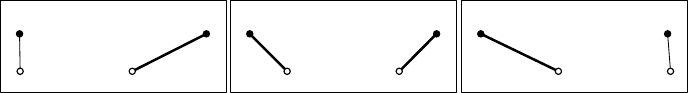}
\caption{A bottleneck matching with different longest edges in different positions of $B$.}
\label{fig:bottleneckEpslon}
\end{figure}


Another observation is that, even for a fixed position of the point sets, there are in general many bottleneck matchings. 
\cref{fig:bottleneck2} shows an example that can be easily generalized to show that a set of $k$ points and a set of $n \geq k$ points can have $(n-1)!/(n-k)!$ different bottleneck matchings. 
Indeed, in the depicted situation only the edge between the points~$b_0$ and $a_0$ is fixed for a matching to be bottleneck, the remaining points can be matched arbitrarily.

\begin{figure}[ht]
\centering
\includegraphics[scale=.8,page=2]{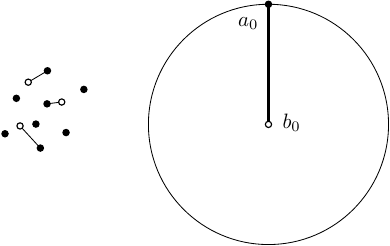}
\caption{A pair of point sets with many bottleneck matchings.}
\label{fig:bottleneck2}
\end{figure}

Moreover, the bottleneck matchings for the point sets in \cref{fig:bottleneck2} all have the same longest edge, and they remain bottleneck matchings if one of the point sets is translated in an open neighborhood of its current position.
However, there may be different edges that are the longest edge of bottleneck matchings everywhere in a neighborhood of a fixed position, as illustrated in \cref{fig:bottleneck3}.
Just like in the previous figures, the bold edges attain the bottleneck value.

A common way to break ties between bottleneck matchings for a given position in order to be more sensitive to the geometry of the point sets is to consider a lexicographic version of bottleneck matchings.
That is, among the matchings whose longest edge is as short as possible, consider those whose second longest edge is as short as possible, and so on. 
For the precise definition, we recall that the \emph{lexicographic order} on $\R^k$ is the total order induced by the relation $(x_1,\dots,x_k) \prec (y_1,\dots,y_k)$ if and only if there exists an $m \in [k]$ such that $x_i=y_i$ for all $i<m$, and $x_m < y_m$.
We write $x \preceq y$ if $ x \prec y$ or $x=y$. 

\begin{figure}
\centering
\includegraphics[scale=.8,page=16]{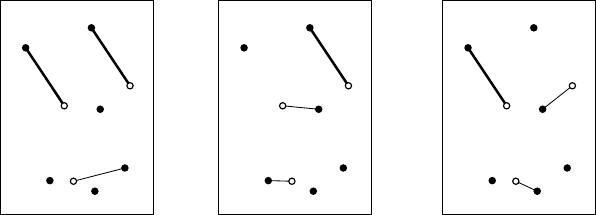}
\caption{Bottleneck matchings with different longest edges.}
\label{fig:bottleneck3}
\end{figure}

\begin{definition}
Let $A,B \subset \R^d$ be two finite point sets with $k=|B| \leq |A|$.
The \emph{lex-bottleneck cost} of a matching $\sigma:B \hookrightarrow A$ is the function $g_\sigma: \R^d \to \R^k$ where the $i$-th coordinate of $ g_\sigma(t)$ corresponds to the length of the $i$-th longest edge of $\sigma$ for $B+t$ and $A$. 
A \emph{lex-bottleneck matching} for $A$ and $B$ is a matching~$\pi$ such that $g_\pi(0) \preceq g_\sigma(0)$ for every other matching $\sigma$. 
\end{definition}

Note that a lex-bottleneck matching is a bottleneck matching as well. 
Although this definition certainly breaks some ties, we see in the next section that it does not guarantee uniqueness (not even in the complement of a nowhere-dense set).

\section{Bottleneck Partial-Matching Voronoi Diagrams}
\label{sect:diagrams}

In this part, we define Voronoi-type diagrams associated with the (lexicographic) bottleneck partial-matching problem under translation.
We discuss basic properties of these structures and derive reasonable bounds on their combinatorial complexity.

\subsection{Definitions and Basic Properties}

As we have seen in the previous section, we face some difficulties when we define a Voronoi-type structure for the bottleneck partial-matching problem.
One of them is the existence of open sets of translations for which neither the bottleneck matching nor the longest edge are uniquely determined. 
The problem of the non-uniqueness of the longest edge can be solved by requiring the point set to be in an ad hoc general position.
The non-uniqueness of the matching may be attacked by considering the lexicographic variant, but even in general position it might fail even in open balls. 
Nevertheless, it is of interest to study the original version as well in order to solve problems like the ones in~\cref{sect:applications} or to explore the minimization diagram of~$\E$.

As expected, the Voronoi-type diagrams we are going to study can be required to be given in form of a polyhedral complex.
This facilitates traversing the partition or optimizing in a region, operations that are often required in related problems. 
In what follows, full-dimensional faces of a polyhedral complex are called \emph{cells}.

\begin{definition}
Let $A,B \subset \R^d$ be finite point sets with $|B| \leq |A|$.
A \emph{bottleneck partial-matching Voronoi diagram} (or \emph{bottleneck diagram}, for short) for $A$ and~$B$ is a polyhedral complex $\TT$ covering $\R^d$ and such that for every cell $C$ of $\TT$ there is at least one matching $\pi_C:B \hookrightarrow A$ such that $f_{\pi_C}(t) \leq f_\sigma(t)$ for all $t \in C$ and all matchings $\sigma:B \hookrightarrow A$. 
A \emph{bottleneck labeling} of a bottleneck diagram is a function mapping each cell to one such matching.
\end{definition}

A coarsening of a bottleneck diagram of the point set in~\cref{fig:pointset} is displayed in \cref{fig:lexdiagram}, where cells with the same label have the same color.
Note that for $B=\{b\}$ the Voronoi diagram of $A-b$ is a bottleneck diagram. 
A diagram for the lexicographic version of the bottleneck cost is defined analogously.

\begin{definition}
Let $A,B \subset \R^d$ be two finite point sets with $|B| \leq |A|$.
A \emph{lex-bottleneck partial-matching Voronoi diagram} (or \emph{lex-bottleneck diagram}, for short) for $A$ and $B$ is a polyhedral complex $\TT$ covering $\R^d$ and such that for every face~$c$ of $\TT$ there is at least one matching $\pi_c:B \hookrightarrow A$ such that $g_{\pi_c}(t) \preceq g_\sigma(t)$ for all~$t$ interior to $c$ and all matchings $\sigma:B \hookrightarrow A$. 
A \emph{lex-bottleneck labeling} of a lex-bottleneck diagram is a function mapping each face to one such matching. 
\end{definition}

\begin{figure}[ht]
\centering
\includegraphics[page=6]{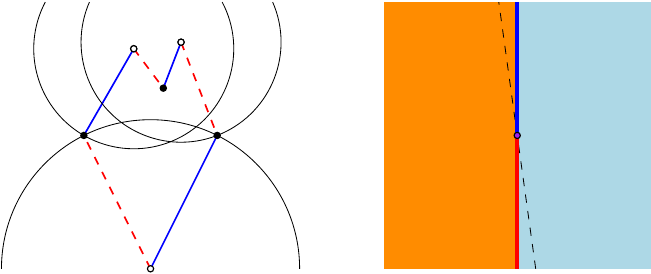}
\caption{A pair of point sets and a lex-bottleneck labeling of the neighborhood of the represented position.}
\label{fig:lexdeg}
\end{figure}

Note that a bottleneck matching in a lower-dimensional face of a bottleneck diagram is given by the labeling of a cell containing it, i.e., the label for a cell is valid everywhere in the cell.
This is not the case for a lex-bottleneck diagram, which is the reason why a label for each face is required. 
\cref{fig:lexdeg} shows an example.
On the left of the figure, the point sets $A$ and $B$ are displayed, for which the blue and the red matchings are both lex-bottleneck matchings. 
On the right, a small neighborhood in a lex-bottleneck diagram around the point $t \in \R^2$ corresponding to the depicted position of the point sets is represented. 
If the point set $B$ (the white dots) is moved infinitesimally to the right, only the blue matching is a lex-bottleneck matching, whereas if it is moved infinitesimally to the left, only the red matching remains lex-bottleneck. 
This forces the cyan and orange regions to be labeled with the blue and red matchings, respectively. 
However, if $B$ is vertically translated by an infinitesimal amount, the longest blue edge and the longest red edge have the same length. 
In addition, the second longest red edge is longer than the second longest blue edge if the perturbation is upwards, while it is shorter if the perturbation is downwards, forcing the blue and red regions to be labeled accordingly. 
For the depicted position (corresponding to the purple point), both matchings are equally good, since the respective shortest edges are equally long.

Since any lex-bottleneck diagram is a bottleneck diagram, we prove some properties for the first, more restrictive type. 
However, later on we devise algorithms that compute a bottleneck labeling more efficiently than a lex-bottleneck one. 
For our applications in~\cref{sect:applications} the first type of labeling is enough. 

\begin{definition}
Given $x,y,v,z \in \R^d$, let
\begin{align*}
h(x,y,v,z)&=\left\{t\in\R^d:\|y+t-x\|=\|z+t-v\|\right\}\\
&=\left\{ t \in \R^d:  2\scalprod{t}{y-x-(z-v)}= \|z-v\|^2 - \|y-x\|^2 \right\}.
\end{align*}
Given two finite point set $A, B \subset \R^d$, let $\lex(A,B)$ be the arrangement of the hyperplanes $h(a,b,a',b')$, called \emph{bisectors}, for all pairs $a,a' \in A$ and $b,b' \in B$ such that $b-a \ne b'-a'$.
\end{definition}

\begin{proposition}\label{epolyhedral}
The arrangement $\lex(A,B)$ is a lex-bottleneck diagram.
\end{proposition}

\begin{proof}
The squared length of an edge matching $b$ to $a=\sigma(b)$ is given by
\[ \| b+t-a \|^2= \|t\|^2 + \|b-a\|^2 + 2\scalprod{t}{b-a}. \]
For a pair of edges $ab,a'b' \in A \times B$, the locus of points $t \in \R^d$ for which $\| b+t-a \|^2=\| b'+t-a' \|^2$ is exactly $h(a,b,a',b')$. 
If $b-a \neq b'-a'$, then this set is a hyperplane, otherwise the whole~$\R^d$. 
Let $c$ be a face in $\lex(A,B)$, and let $\pi$ be a lex-bottleneck matching for a point $t_0$ in the relative interior of~$c$. 
A lex-bottleneck matching for the translation $t\in\R^d$ only depends on the relative length of the possible edges in $A\times (B+t)$.
Hence, the matching $\pi$ remains lex-bottleneck for $t$ as well as long as no edge becomes strictly shorter than another edge that was strictly longer for $t_0$. 
By definition of $\lex(A,B)$, this cannot happen in the relative interior of~$c$.
Since an arrangement of hyperplanes is a polyhedral complex, $\lex(A,B)$ is a lex-bottleneck diagram for $A$ and~$B$. 
\end{proof}

Observe that there may be open sets for which two different matchings are lex-bottleneck matchings, as shown in \cref{fig:lexbotcycle}. 
This is because every edge from the red matching can be paired with an edge of the blue matching having the same length for any position of the matching. 
As long as the blue match and the red match of every point in $B$ are its two closest points (as in the three positions represented in the figure), both are lex-bottleneck matchings. 
Nonetheless, we see now that the \emph{matched sets} $\sigma(B)$ of such matchings $\sigma:B\hookrightarrow A$ coincide.

\begin{figure}
\centering
\includegraphics[page=4]{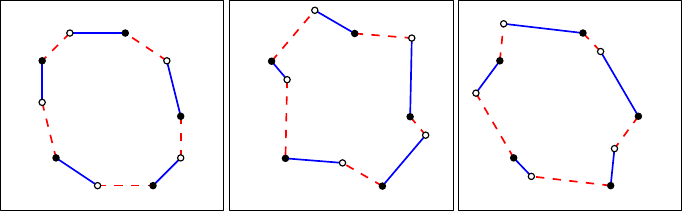}
\caption{Two lex-bottleneck matchings, represented for three positions of $B$.}
\label{fig:lexbotcycle}
\end{figure}

\begin{proposition}
Let $A,B \subset \R^d$ be finite point sets with $|B| \leq |A|$.
If two matchings are lex-bottleneck matchings in an open set $U \subset \R^d$, then they have the same matched set and they have the same lex-bottleneck cost for any $t \in \R^d$. 
In particular, in the interior of a cell of $\lex(A,B)$ there is a unique subset $A' \subseteq A$ that is the matched set of all lex-bottleneck matchings for that cell.
\end{proposition}

\begin{proof}
As argued in the proof of~\cref{epolyhedral}, the set of lex-bottleneck matchings is the same for all the translations interior to a cell of $\lex(A,B)$. 
Let $C$ be a cell of $\lex(A,B)$, $t_0$ be a point interior to $C$, and $\pi$ and $\sigma$ be two different lex-bottleneck matchings for $B+t_0$ and~$A$. 
We regard here the matchings $\pi$ and $\sigma$ as sets of edges on the vertex set $A \cup B$. 
Consider then the symmetric difference of $\pi$ and $\sigma$, i.e., the set of edges that belong to $\pi$ or to $\sigma$ but not to both.
The graph with this set of edges is a collection of (even length) vertex-disjoint paths and cycles, whose edges from $\pi$ alternate with edges from $\sigma$, since the edges of a matching are pairwise-disjoint (see~\cref{subsect:notbipgraphs} for details). 

In addition, since we assume that both matchings are lex-bottleneck matchings, we have that there is a one-to-one correspondence between edges of every path or cycle belonging to $\pi$ and the edges from the same path or cycle belonging to $\sigma$, such that the corresponding edges have the same length.
Indeed, assume for a contradiction that there is no such correspondence for some path or cycle $\gamma$. 
Then, the restriction of one of the matchings, say $\pi$, to $\gamma$ would be better than the restriction of the other (in the lexicographic sense).
The result of replacing in $\sigma$ the edges of $\sigma \cap \gamma$ with the edges of $\pi \cap \gamma$ is a matching and it is lexicographically better than $\sigma$, which contradicts its assumed optimality. 

It follows from the proof of~\cref{epolyhedral}, that the edges $ab$ and $a'b'$ have the same length over an open set if and only if $b-a=b'-a'$. 
Thus, there is no path in the symmetric difference, because following the edges in a path we would arrive to the starting vertex, since every edge is ``cancelled'' by the corresponding edge of the other matching. 
Therefore, the symmetric difference is made exclusively of cycles, and hence the matched sets of $\pi$ and $\sigma$ coincide. 
\end{proof}

\subsection{Complexity of the Bottleneck Diagrams}

The construction in the proof of \cref{epolyhedral} leads immediately to a first bound on the complexity of a lex-bottleneck diagram.

\begin{corollary}\label{cor:firstbound}
For any pair of point sets $A,B \subset \R^d$ with $k=|B|\leq|A|=n$, there is a lex-bottleneck diagram of combinatorial complexity $O(n^{2d}k^{2d})$.
\end{corollary}

\begin{proof}
By \cref{epolyhedral} the arrangement $\lex(A,B)$ is a lex-bottleneck diagram for $A$ and~$B$.
It is well-known (cf.~\cite{Edelsbrunner}) that the complexity of an arrangement of~$m$ hyperplanes in $\R^d$ is $O(m^d)$. 
The arrangement $\lex(A,B)$ consists of $ \binom{n}{2}\binom{k}{2}=O(n^2k^2)$ hyperplanes in $\R^d$, and hence the claimed bound follows.
\end{proof}

Our aim is now to show that there are many hyperplanes of $\h(A,B)$ that we can safely ignore.
This allows for an improvement of the above bound for planar point sets.
We need to introduce some necessary notation and collect some auxiliary results before we can present our argument.

First of all, we state some simple results on hyperplane arrangements, the first of which can be found, e.g., in~\cite{Edelsbrunner}, and will be used implicitly later on.

\begin{proposition}
For a finite point set $P \subset \R^d$, consider the linear functions \[h_p(x)=\|p\|^2 -2 \scalprod{x}{p}, \text{ for } p \in P.\]
For a fixed $k \in [|P|]$, the $k$-level of the arrangement of hyperplanes associated to these functions projects onto the order-$k$ Voronoi diagram of $P$. \end{proposition}

\begin{proposition}\label{prop:VorArran}
Let $S \subset \R^2$ be a finite point set and let $\mathcal{Z}$ be the set of planes
\[ z_p=\left\{\left(x,y,\|p\|^2 -2 \scalprod{(x,y)}{p}\right) : (x,y) \in \R^2 \right\} \subset \R^3, \text{ for } p \in S. \]
Every three planes in $\mathcal{Z}$ intersect in at most one point. 
Equivalently, the locus of points equidistant from three different points of $S$ is either a point or empty.
\end{proposition}

\begin{proof}
For every point $q \in \R^2$, the vertical order of the planes $z_p$ over $q$ is the same as the distances from~$q$ to the corresponding $p\in S$. 
The locus of points equidistant from three points in the plane is the center of the circle through them, or the empty set if they are collinear.
Hence, three planes can coincide in at most one point.
\end{proof}

The main ingredient to improve the complexity bound in \cref{cor:firstbound} is the following lemma.

\begin{lemma}[Ben-Avraham et al.~\cite{ESA}]
\label{ClarkShor}
Let $A \subset \R^2$ be a set of $n$ points. 
There are $O(nk)$ bisectors that support all edges of order-$j$ Voronoi diagrams of $A$ for all $j \leq k$.
\end{lemma}

Finally, we fix some technicalities for lex-bottleneck matchings.

\begin{definition}\label{def:minimal_geom}
Let $A,B \subset \R^d$ be two finite point sets with $k=|B| \le |A|$.
Given $b \in B$, a set of edges $M_b\subseteq A\times B$ is called a \emph{$b$-minimal set} if $|M_b|=k$ and no edge $ab \in (A\times B) \setminus M_b$ is strictly shorter than any edge in $M_b$. 
The \emph{candidate set} of $b\in B$ at a position $t\in\R^d$ is the set $E_b(t)\subseteq A\times(B+t)$ with $|E_b(t)| \geq k$ and such that every edge $(a,b+t) \notin E_b(t)$ is strictly longer than every edge in~$E_b(t)$.
Furthermore, we write $Z(t)=\cup_{b\in B} E_b(t)$.
\end{definition}

\begin{lemma}\label{lem:minimal_geom}
Let $A,B \subset \R^d$ be two finite point sets with $k=|B| \le |A|$.
\begin{enumerate}[(i)]
\item If a subset $M \subseteq A\times B$ contains a $b$-minimal set for each $b \in B$, then~$M$ contains a lex-bottleneck matching for $A$ and $B$.\label{lem:minimal1_geom}
\item Every lex-bottleneck matching for $A$ and $B$ is contained in the union~$Z=Z(0)$ of candidate sets.\label{lem:minimal2_geom}
\end{enumerate}
\end{lemma}

\begin{proof}
(i): Let $\mu \subset A\times B$ be a lex-bottleneck matching. 
If $\mu \subseteq M$, nothing is left to prove.
Otherwise, let $ab \in A\times B$ be an edge in $\mu \setminus M$, and let $M_b \subseteq M$ be a $b$-minimal set. 
Since $\mu$ matches exactly $k$ points of $A$ and $ab \not \in M \supseteq M_b$, there must be an edge $a'b \in M_b \setminus \mu$. 
The matching $(\mu \setminus \{ab\}) \cup \{a'b\}$ is lexicographically at least as good as $\mu$ and it uses one more edge of~$M$. 
Repeating the process, we end up with a lex-bottleneck matching contained in~$M$.

(ii): Assume that $\mu$ is a lex-bottleneck matching for $A$ and $B$ and that it contains an edge $ab \notin Z$. 
Let $a'b$ be an edge in $E_b(0) \setminus \mu \ne \emptyset$. 
By definition of candidate sets, the matching $(\mu \setminus \{ab\}) \cup \{a'b\}$ is lexicographically better than~$\mu$, contradicting its optimality.
\end{proof}

Now we are well-prepared to prove our main result of this section.

\begin{theorem}\label{thm:bot} \label{arrangementrel}
For any pair of point sets $A,B \subset \R^2$ with $k=|B|\leq|A|=n$, there is a lex-bottleneck diagram of complexity $O(n^2k^6)$.
\end{theorem}

\begin{proof}
Observe that $E_b(t_1)=E_b(t_2)$ for every pair of points $t_1,t_2$ interior to a cell~$C$ of $\lex(A,B)$, for all $b \in B$. 
Therefore, we write $E_b(C)=E_b(t)$ and $Z(C)=Z(t)$ for any point $t$ interior to a cell $C$ of $\lex(A,B)$.
By~\cref{lem:minimal_geom}-(\ref{lem:minimal2_geom}), any lex-bottleneck matching for a point interior to $C$ is contained in $Z(C)$. 
Let now $t_0$ be a point in a lower-dimensional face $L$ of $\lex(A,B)$, and let $C \supset L$ be a cell. 
For continuity reasons, the set $E_b(C) \subseteq E_b(t_0)$ is a $b$-minimal set at~$t_0$ as well. 
Therefore, by ~\cref{lem:minimal_geom}-(\ref{lem:minimal1_geom}), there is a lex-bottleneck matching for~$t_0$ contained in $Z(C)$. 

Now, consider a labeling $\Lambda$ of $\lex(A,B)$ that labels every face of a cell $C$ with a matching contained in $Z(C)$.
We say that an edge $F$ of $\lex(A,B)$ between two cells $C_l$ and $C_r$ \emph{uses} a bisector $h=h(a,b,a',b')$ if $F \subset h$ and $ab,a'b' \in Z(C_l) \cup Z(C_r)$.
If no edge uses a bisector $h$, then~$h$ can be omitted from the arrangement $\lex(A,B)$ while maintaining its property of being a lex-bottleneck diagram. 
Each new face in the resulting hyperplane arrangement is a union of a set of old faces, i.e., before removing $h$, and can be labeled according to the label of any of them, since $h$ does not intersect any face whose label in~$\Lambda$ uses the edges $ab$ or $a'b'$. 
We show now that many bisectors are not used by any edge, distinguishing the two following cases. 

Let $h(a,b,a',b)$ be a bisector used by an edge $F=C_l \cap C_r$, and consider the point set $S(b)=A-b$. 
The edge $F$ must be contained in an order-$j$ Voronoi edge of $S(b)$ for some $j \le k$.
Indeed, in view of~\cref{prop:VorArran}, infinitesimally to the right or to the left of (a point in the relative interior of) $F$ both $a-b$ and $a'-b$ are among the $k+1$ closest points.
In addition, for a point $t_0$ in the relative interior of~$F$, the points $a-b$ and $a'-b$ are the only two points of $S(b)$ that lie at distance $\|b+t_0-a\|=\|b+t_0-a'\|$. 
Hence, there is a circle centered at $t_0$ and through $a-b$ and $a'-b$ that contains $j-1 \le k-1$ points in its interior, which is a characterization for points in the relative interior of edges of the order-$j$ Voronoi diagram.  

Let $h(a,b,a',b')$ with $b \ne b'$ be a bisector used by an edge $F=C_l \cap C_r$, and consider the point set ${S(b,b')=(A-b) \cup (A-b')}$. 
Note that the number of points in $S(b,b')$ is not necessarily $2n$: a point $a_1-b=a_2-b'$ can belong to both $A-b$ and $A-b'$. 
However, this is the case if and only if the edges $a_1b$ and $a_2b'$ are equally long everywhere. 
In particular, the points $a-b$ and $a'-b'$ are distinct since, otherwise, they would not induce any bisector. 
Furthermore, simple algebraic manipulations show that $t \in \R^2$ is closer to $a_1-b$ than to $a_2-b'$ if and only if $b+t$ is closer to~$a_1$ than $b'+t$ is to $a_2$, for any choice of $a_1,a_2\in A$. 
Since the bisector is used by $F$, the point $a-b$ is among the $k$ closest points of $A-b$ infinitesimally to at least one of the sides of~$F$. 
\cref{prop:VorArran} ensures that along the interior of $F$ only the points $a-b$ and $a'-b'$ are at distance $\|b+t_0-a\|=\|b'+t_0-a'\|$ among the points in $S(b,b')$, which implies that in fact $a-b$ is among the $k$ closest points of $A-b$ infinitesimally to both sides of~$F$. 
Similarly, the point $a'-b'$ belongs to the $k$ closest points of $A-b'$ for points infinitesimally away from $F$. 
Hence, for any point $t_0$ in the relative interior of~$F$, there is a disk centered at $t_0$ and passing through $a-b$ and $a'-b'$ that contains $j-1 \le 2k-2$ points of $S(b,b')$ in its interior. 
Equivalently, the bisector $h(a,b,a',b')$ supports an order-$j$ Voronoi edge of $S(b,b')$ for some $j \le 2k-1$.

Applying \cref{ClarkShor} to $S(b)$, for all $b \in B$, and to $S(b,b')$, for every pair $b,b' \in B$, it follows that the number of bisectors that are used by some edge is $O(k^2 \cdot nk)$. 
The complexity of the diagram resulting from removing all unused bisectors from $\lex(A,B)$ is thus $O(n^2k^6)$. 
%
\end{proof}

Based on a one-dimensional example of Rote~\cite{Rote2010}, one may derive a lower bound on the complexity of any lex-bottleneck diagram.
The proof is a particular case of the result for stable matchings that can be found in~\cite{ESA_journal}.

\begin{proposition}\label{lowerboundStable}
For any $k,n \in \N$ with $n \ge k \ge d$, there exist point sets $A,B \subset \R^d$ with $|B|=k$ and $|A|=n$ such that any lex-bottleneck diagram for $A$ and $B$ has complexity ${\Omega(k^d (n-k)^d)}$.
\end{proposition}

In particular, this shows that \cref{thm:bot} is optimal with respect to the size of the bigger set~$A$.

\section{Construction of the Bottleneck Diagrams}
\label{sect:construction}

In this section we are concerned with construction algorithms for the bottleneck diagrams introduced before.
The basic idea is to first construct the reduced hyperplane arrangement discussed in the proof of \cref{thm:bot}, and then traverse it while computing a bottleneck labeling for each cell (or face) of the structure.
The latter is based on the well-developed theory of bottleneck assignments in weighted bipartite graphs whose concepts and methods that are relevant for our purposes we introduce first.

\subsection{Notation and Techniques for Matchings in Weighted Bipartite Graphs}
\label{subsect:notbipgraphs}

As usual, let $A,B\subset\R^d$ be finite point sets with $|B|\leq|A|$.
The problem of finding a bottleneck matching for a fixed position of~$B$ can be translated into a matching problem in a weighted bipartite graph on $A$ and $B$, where the weight of an edge from $ab \in A\times B$ is the Euclidean distance between the corresponding points.
Most of the geometric definitions from \cref{sect:botmatchings,sect:diagrams} have graph-theoretic analogs.
Note that we use the same symbols and names for corresponding concepts in the graph setting.

Let $G=(A,B;E)$ be a bipartite graph with edge set $E \subseteq A\times B$ and vertex set partitioned into the components $A$ and $B$.
A \emph{matching} in $G$ is a set $\sigma \subseteq E$ such that every vertex in $A \cup B$ is incident to at most one edge of $\sigma$. 
As in the geometric setting, we identify a matching with the injection from~$B$ into $A$ it induces, and we simplify notation by denoting an edge $(a,b)\in E$ by $ab$.
A~\emph{maximum matching} is a matching of maximum cardinality. 
Vertices that belong to an edge of a matching are called \emph{matched vertices} and otherwise \emph{exposed vertices}.

An \emph{alternating path/cycle} for a matching $\sigma$ is a path/cycle in $G$ with no repeated vertices such that the even edges are in $\sigma$ and the odd ones are in~${E\setminus\sigma}$.
An \emph{augmenting path} for~$\sigma$ is an alternating path starting and ending at exposed vertices.
Note that, if~$\gamma$ is an augmenting path for $\sigma$, the matching ${\sigma \oplus \gamma=(\sigma \setminus \gamma) \cup  (\gamma \setminus \sigma)}$ has one more matched vertex than~$\sigma$.
In general, for two sets of edges $\sigma,\tau \subseteq E$ the connected components of the graph induced by $\sigma \oplus \tau$ are called its \emph{components}.
If $\sigma$ and $\tau$ are maximum matchings, the components of $\sigma \oplus \tau$ are paths or cycles, since every vertex has degree at most two in $\sigma \oplus \tau$. 

\subsubsection*{Bottleneck Assignments in Bipartite Graphs}

Given a bipartite graph $G=(A,B;E)$, we let $w:E \to \R_{{\geq}0}$ be a function giving weights to its edges. 
The bottleneck cost of a matching $\pi$ in $G$ with respect to $w$ is the maximum $w$-value attained by the edges in~$\pi$.
The problem of finding a maximum matching of minimum bottleneck cost, henceforth referred to as a \emph{bottleneck matching}, for the complete and balanced case, i.e., $E=A \times B$ and $k=n$, has been widely studied in the last decades under the name of the \emph{bottleneck assignment problem}.
The most prominent approaches for this problem are the \emph{threshold methods} and the \emph{augmenting path methods}. 
Details and related studies can be found in the book dedicated to assignment problems from Burkard, Dell'Amico \& Martello~\cite{BurkardBook}. 
The threshold algorithms conduct a binary search on the possible values for the edge with maximum weight of a bottleneck matching. 
At each stage, the edges with bigger weight than the threshold are ignored, and a maximum matching computation in the modified graph is performed. 
One of the best-known algorithms to find a maximum matching in a bipartite graph is due to Hopcroft \& Karp~\cite{HK}.
It runs in $O\big(|E| \sqrt{ \nu(G)}\big)$ time, where $\nu(G)$ is the size of the maximum matchings in $G$.
The algorithm by Alt, Blum, Mehlhorn \& Paul~\cite{AltBlum} finds a maximum matching more efficiently if the graph is ``dense''. 
This fact was exploited by Punnen \& Nair~\cite{Punnen} to develop an alternative algorithm for the bottleneck assignment problem. 
The pure threshold method is preferable for dense graphs, the method from Gabow \& Tarjan~\cite{Gabow2} is better for sparse graphs and the approach of Punnen \& Nair covers the range in between. 

Similar to the geometric situation, a variant of the bottleneck assignment problem is the \emph{lexicographic bottleneck assignment problem}, introduced in \cite{Burkard} and revisited in \cite{Sokkalingam}. 
For this problem, the cost of a matching $\sigma:B \hookrightarrow A$ is the result of sorting decreasingly the values $w(\sigma(b)b)$, for all $b \in B$.
A \emph{lexicographic bottleneck matching} is a matching minimizing the cost, when the corresponding cost vectors are compared lexicographically. 
Note that a lexicographic bottleneck matching is necessarily a bottleneck matching.  

In the unbalanced case, as long as the graph is complete, the maximum matchings have size~$|B|$. 
Let us recall the concepts of $b$-minimal and candidate sets from \cref{def:minimal_geom} by reviewing their definition in the graph setting.

\begin{definition}\label{def:minimal}
Let $G=(A,B;E)$ be a bipartite graph with $k=|B| \le |A|$, and $w:E \to \R_{{\geq}0}$ be a function giving weights to the edges. 
Given $b \in B$, a set $M_b\subseteq E$ is called a \emph{$b$-minimal set} if $|M_b|=k$ and no edge $ab \in E \setminus M_b$ has strictly smaller weight than any edge in $M_b$. 
The \emph{candidate set of $b$} is the set $E_b\subseteq E$ with $|E_b| \ge k$ and such that every edge $ab \in E \setminus E_b$ has strictly larger weight than any edge in $E_b$.
\end{definition}

\cref{lem:minimal_geom} holds analogously for bipartite graphs and can be considered as a property similar to Hall's marriage theorem. 
As a consequence thereof, we can select a set of~$k^2$ edges which is guaranteed to contain a lexicographic bottleneck matching. 
This pruning of the graph can be done in $O(n k)$ time using selection algorithms.
Although we do not know whether the graph will be dense or sparse after pruning the non-relevant edges and isolated vertices, we have that both $|B|+|A|$ and $|E|$ are $O(k^2)$. 
Thus, the best worst-case running time for our scenario is provided by the algorithm of Gabow \& Tarjan, which runs in $O(k^{2} \sqrt{k \log k})$ time, according to the analysis in \cite{BurkardBook}. 
The approach in \cite{Sokkalingam}, based on solving a sequence of linear sum assignment problems and bottleneck assignment problems, boils down to an algorithm for the computation of a lexicographic bottleneck matching running in $O(k^4)$ time after the pruning. 

We close this part by stating some well-known results in the matching literature, which are used by the aforementioned algorithms and required below.

\begin{lemma}[Berge's Lemma \cite{Berge}] \label{lem:berge}
A matching $\sigma$ is a maximum matching in a bipartite graph $G$ if and only if there is no augmenting path for $\sigma$ in $G$. 
\end{lemma}

\begin{definition}
Let $G=(A,B;E)$ be a bipartite graph and $w:E \to \R_{{\geq}0}$ be a function giving weights to its edges. 
Given $r \in [|E|]$, let
\[ E(r)= \left\{e \in E: w(e) \text{ is among the } r \text{ smallest values of } w(E) \right\}. \]
We denote by $G(r)$ the graph $G=(A,B;E(r))$.
\end{definition}

We say that a matching in a bipartite graph is a \emph{complete matching} if it matches all the points in the smaller set of vertices.

\begin{proposition}\label{lemabot}
Let $G=(A,B;E)$ be a bipartite graph that has a complete matching and let $w:E \to \R_{{\geq}0}$ be a function giving weights to the edges of~$G$.
\begin{enumerate}[(i)]
\item If $G(r)$ has a complete matching, then $G(j)$ has a complete matching for all $j>r$.
\item A complete matching in $G(r)$ is a bottleneck matching if and only if $G(r-1)$ has no complete matching. \label{characbot}
\end{enumerate}
\end{proposition}

\subsection{Main Construction Theorem}

In this section, we discuss some algorithmic techniques in order to construct a labeled bottleneck and a lex-bottleneck diagram for a pair of point sets in the plane. 
To this end, we introduce the following notation.

\begin{definition}\label{def:lexo}
For finite point sets $A,B \subset \R^2$, we denote by $\lexo(A,B)$ the arrangement constructed in the course of the proof of~\cref{arrangementrel}. 
Given a cell $C$ of $\lexo(A,B)$, we denote by $Z(C)$ the set $Z(t)$ for any point $t$ interior to $C$.  
\end{definition}

We first describe how to use established techniques in order to construct an unlabeled lex-bottleneck diagram.

\begin{lemma}\label{lexarrange}
The lex-bottleneck diagram $\lexo(A,B)$ for point sets $A,B \subset \R^2$ with $k=|B|\leq |A|=n$ can be constructed in $O(n^2k^6)$ time. 
\end{lemma}

\begin{proof}
An arrangement of $m$ lines in the plane can be computed in $O(m^2)$ time using an optimal algorithm, such as the incremental algorithm~\cite{ChazelleGuibas} or a topological sweep~\cite{EdelsbrunnerGuibas}. 
However, the proof of~\cref{thm:bot} is not constructive and, hence, it is not obvious how to select the bisectors that are used by some edge (in the sense of~\cref{arrangementrel}) among the $O(n^2k^2)$ candidates. 
Fortunately, the algorithm by Chan~\cite{Chan} constructs the facial structure of the $({\leq} s)$-level of an arrangement of $m$ planes in $O\!\left(m \log m+ms^2\right)$ expected time. 
In addition, this algorithm can be derandomized, leading to a deterministic version running in $O\!\left(ms^2(\log m / \log s)^{O(1)}\right)$ time. 
We can then construct the $O(k^2)$ necessary structures in $O\!\left(k^4 n ( \log n / \log k)^{O(1)}\right)$ and traverse each of them discovering the $O(nk^3)$ used bisectors to finally construct their arrangement.
\end{proof}

%
%
%

We now show how to find a bottleneck labeling of $\lexo(A,B)$. 
Before detailing the algorithm, we need a technical lemma that examines how small changes in a graph affect its bottleneck matchings. 
Recall that a matching in a bipartite graph on $A$ and $B$ is a complete matching if it matches all the points of~$B$.

\begin{lemma}\label{lem:bot}
Let $G=(A,B;E,w)$ be a bipartite graph with $w:E \to [|E|]$ giving weights to its edges.
Let $\mu$ be a bottleneck matching for $G$, and let $l \in E$ be the longest edge of $\mu$ in $G$. 
For a fixed $j \in [|E|-1]$, let $G'=(A,B;E,w')$ be the weighted bipartite graph, where $w'$ coincides with $w$ except that $w'(e)=j$ if $w(e)=j+1$, and $w'(e)=j+1$ if $w(e)=j$, for all edges $e \in E$.
\begin{enumerate}[(i)]
\item If $w(l) \not \in \{j,j+1\}$, then $\mu$ is a bottleneck matching for $G'$. 
\item If $w(l) \in \{j,j+1\}$ and $G'(j)$ does not have a complete matching, then $\mu$ is a bottleneck matching for $G'$. 
\item If $w(l) \in \{j,j+1\}$ and $G'(j)$ has a complete matching $\nu$, then $\nu$ is a bottleneck matching for $G'$. 
\end{enumerate}
\end{lemma}

\begin{proof}
Note first that $G(i)=G'(i)$ for all $i \neq j$ and recall the characterization of bottleneck assignments in~\cref{lemabot}-(\ref{characbot}). 

(i): Let $w(l)=s \not \in \{ j,j+1\}$. 
Since $\mu$ is a bottleneck matching for $G$, the graph $G(s)$ has a complete matching and $G(r)$ does not, for any $r < s$. 
Hence, the graph $G'(s)=G(s)$ has a complete matching and $G'(s-1)=G(s-1)$ does not.
In addition, $\mu \subseteq G'(s)$, which ensures that $\mu$ is indeed a bottleneck matching for~$G'$.

(ii): Since we assumed that $G'(j)$ has no complete matching, any complete matching in $G'(j+1)$ is a bottleneck matching for~$G'$. 
The matching $\mu$ is contained in $G'(j+1)=G(j+1)$ because we assumed $w(l) \in \{j,j+1\}$. 

(iii): It is clear that $G'(i)= G(i)$ for all $i<j$ and, hence, it does not have a complete matching. 
Since $\nu \subseteq G'(j)$, it is a bottleneck matching for $G'$.
\end{proof}

As we have seen in \cref{sect:botmatchings}, several edges in $A\times B$ can have the same length wherever the point set $B$ is translated.
However, this happens if and only if all such edges are between points $b\in B$ and $a\in A$ with the same vector~$b-a$.
In order to also handle point sets in this special position, we introduce the following equivalence relation.

\begin{definition}\label{def:classes}
Let $A,B\subset\R^2$ be finite point sets. Two edges $ab,a'b' \in A \times B$ are said to be \emph{equivalent} if $b-a = b'-a'$. 
\end{definition}

Note that any two equivalent edges match distinct elements of $B$ and, hence, the size of every equivalence class is at most $k$. 

\begin{algorithm}\label{alg:construction}
 \SetAlgoNoEnd
 \DontPrintSemicolon
 \KwIn{planar points sets $A,B \subset \R^2$ with $k=|B|\leq|A|=n$}
 \KwOut{labeled bottleneck diagram for $A$ and $B$}
 \smallskip
 sort $\{b-a : b\in B, a\in A\}$ lexicographically; group in equivalence classes\;\nllabel{line_grouping}
 \ForEach{$b\in B$}{\nllabel{line_startL}
  compute $(\leq k)$-level of arrangement for $S(b)=A-b$\;\nllabel{line_chanb}
  \ForEach{used bisector $h=h(a,b,a',b)$}{
   $\EP(h)\leftarrow\EP(h)\cup\left\{\left(ab,a'b\right)\right\}$\;\nllabel{line_epb}
  }
  \ForEach{$b'\in B\setminus\{b\}$}{
   compute $(\leq 2k-1)$-level for $S(b,b')=(A-b)\cup(A-b')$\;\nllabel{line_chanbb}
   \ForEach{used bisector $h=h(a,b,a',b')$}{
    $\EP(h)\leftarrow\EP(h)\cup\left\{\left(ab,a'b'\right)\right\}$\;\nllabel{line_epbb}
   }
  }
 }
 $\lexo(A,B)\leftarrow$ arrangement of the used bisectors\;\nllabel{line_endL}
 $C\leftarrow$ a cell of $\lexo(A,B)$; $t\leftarrow$ an interior point of $C$\;
 $S\leftarrow$ sorted list of $\{\|b+t-a\|:b\in B, a\in A\}$\;
 $G\leftarrow$ graph with edges $Z(C)$ weighted by their order $w$ in $S$\;
 label $C$ with $\mu\leftarrow$ bottleneck matching in $G$ with respect to $w$\;\nllabel{line_firstBM}
 \While{there is an unprocessed neighboring cell $D$ of $C$}{
  $h\leftarrow D\cap C$\;
  \ForEach{$(ab,a'b)\in\EP(h)$}{
   \If{$(ab\in G, a'b\notin G)$ {\bf or} $(ab\notin G, a'b\in G)$}{
    $G\leftarrow G\setminus\{\textrm{contained edge}\}\cup\{\textrm{non-contained edge}\}$\;
    $l\leftarrow$ longest edge of $\mu$\;
    $\mu\leftarrow$ augmented matching $\mu\setminus\{\textrm{contained edge}\}$ in $G(w(l))$\;
   }
   \If{$ab\in G$ {\bf and} $a'b\in G$}{
    swap $w(ab)$ and $w(a'b)$ in $G$\;
   }
  }
  \ForEach{$(ab,a'b')\in\EP(h)$ with $b\neq b'$}{
   swap $w(ab)$ and $w(a'b')$ in $G$\;
   $l\leftarrow$ longest edge of $\mu$; $j\leftarrow\min\{w(ab),w(a'b')\}$\;
   \If{$w(l)\in\{j,j+1\}$ {\bf and} $G(j)$ has a complete matching $\nu$}{
    $\mu\leftarrow\nu$\;
   }
  }
  label $D$ with $\mu$\;
  $C\leftarrow D$\;\nllabel{line_laststep_last}
 }
 \caption{LabeledBottleneckDiagram($A$,$B$)}
\end{algorithm}

\begin{theorem}
\label{thm:constr}
Let $A,B \subset \R^2$ be with $k=|B|\leq|A|=n$. 
A labeled bottleneck diagram of $A$ and~$B$ can be computed in $O(n^2k^{8})$ time, and a labeled lex-bottleneck diagram in $O(n^2k^{10})$ time. 
\end{theorem}

\begin{proof}
We construct the diagrams by labeling the cells (and faces) of $\lexo(A,B)$ with a (lex-)bottleneck matching. 
A naive algorithm to do this would compute such a matching from scratch in every cell. 
However, we can maintain a bottleneck matching during a traversal of $\lexo(A,B)$, improving the time complexity of the algorithm.
Unfortunately, this is not the case for the lex-bottleneck diagram, for which the best algorithm we know recomputes (most of) the matching in a number of faces of $\lexo(A,B)$ that we are not able to bound away from its total complexity.

We detail first the algorithm to construct a bottleneck labeling.
A pseudocode description is provided in \cref{alg:construction}.
We start by grouping the edges $b-a$, for $b\in B$ and $a\in A$, into equivalence classes, i.e., groups of edges that have the same length for any fixed translation, as defined in \cref{def:classes}.
The involved sorting can be done in $O(nk \log n )$ time. 

Then, we construct the line arrangement $\lexo(A,B)$ as described in \cref{lexarrange}, and we remember the involved edges for every used bisector. 
Note however that a bisecting line might be selected several times during this process.
That is why we record a list of all the pairs of edges $(ab,a'b')$ that induce such a bisector (see Lines~\ref{line_epb} and~\ref{line_epbb}).
As a consequence of~\cref{prop:VorArran}, every such pair of edges inducing a fixed bisector of $S(b)$ is counted by~\cref{ClarkShor}. 
That is, if $h(a_1,b,a_2,b)=h(a_3,b,a_4,b)$, the first pair is counted as an order-$j_1$ edge and the second pair as an order-$j_2$ edge of $S(b)$ with $j_1 \ne j_2$. 
More precisely, since we could infinitesimally perturb the points in $S(b)$ such that $h(a_1,b,a_2,b) \ne h(a_3,b,a_4,b)$ for any choice of different points $a_1,a_2,a_3,a_4 \in A$ without altering the level of the edges inducing them, the bound in~\cref{ClarkShor} counts already all the pairs inducing the same bisector. 
Every point in $S(b,b')$, for $b,b' \in B$ with $b \ne b'$, can correspond to two equivalent edges. 
Hence, for each ``double'' point inducing a bisector~$h$, we add the corresponding additional pair of edges to the list $\EP(h)$ (see Line~\ref{line_epbb}). 
Therefore, the total number of edge pairs associated to bisectors is only a constant factor bigger than the bound on the total number of bisectors defining $\lexo(A,B)$. 
This shows that there is no overhead in the running time to record the edge pairs and so the Lines~\ref{line_startL} to~\ref{line_endL} take $O(n^2k^6)$ time.

The final step of the algorithm is to traverse the just constructed arrangement $\lexo(A,B)$ and to find a bottleneck matching for each of its cells.
We first initialize the traversal in an arbitrarily chosen cell $C$ of $\lexo(A,B)$.
To this end, we pick a point $t$ interior to $C$, e.g., the centroid of its vertices, and we sort the values $\|b+t-a\|$ choosing one representative edge $ab \in A \times B$ from every equivalence class. 
We initialize also a graph~$G$ with the~$k^2$ edges of $Z(C)=\bigcup_{b\in B}E_b(t)$, since we know by the graph theoretic version of~\cref{lem:minimal_geom}-(\ref{lem:minimal2_geom}) that it contains a, and in fact every, bottleneck label for $C$. 
Moreover, we construct the weight function $w:Z(C) \to [k^2] $ for $G$ representing the order of the lengths of the edges of $Z(C)$ in the relative interior of $C$. 
We then find in Line~\ref{line_firstBM} a bottleneck matching $\mu$ in $G$ in $O(k^{2} \sqrt{k \log k})$ time using the Gabow-Tarjan algorithm introduced in \cref{subsect:notbipgraphs}. 

After this initialization, we now traverse $\lexo(A,B)$ while maintaining the graph~$G$ such that for the current cell $D$ it has edges $Z(D)$ and weight function $w:Z(D) \to [k^2]$ encoding the relative lengths of these edges. 
This is done in Lines~\ref{line_firstBM} to~\ref{line_laststep_last}.
Let $h$ be the current bisector to be crossed. 
We first consider the edge pairs of the type $(ab,a'b)$ for~$h$.
If none of the edges belong to $G$, there is nothing to do. 
If both edges belong to $G$,~\cref{prop:VorArran} ensures that they have consecutive weights which need to be swapped during the crossing of the bisector. 
The current bottleneck matching~$\mu$ is not affected.
If exactly one of the edges, say $a'b$, does not belong to the graph $G$, we include it and remove~$ab$, because it is not part of the candidate set anymore. 
By~\cref{prop:VorArran}, the weight of the new edge is the same as the old one. 
In addition, if the removed edge $ab$ belongs to the current bottleneck matching~$\mu$ with longest edge $l$, then a new matching having longest edge of weight $w(l)$ can be found in the updated graph.
In other words, a pair involving edges incident to the same $b \in B$ may change the candidate set replacing an edge by another edge with the same weight, but the weight of the longest edge of the bottleneck matchings remains invariant. 
In order to find the new bottleneck matching, it is enough to augment $\mu \setminus \{ab\}$ in the graph $G(w(l))$. 
Such an augmenting path is guaranteed to exist by~\cref{lem:berge} because $G(w(l))$ has a complete matching. 

Consider now all the edge pairs of the type $(ab,a'b')$, with $b \ne b'$, in the list of~$h$.  
If at least one of the edges is not contained in $G$, there is nothing to do. 
Otherwise, the weights of the edges need to be swapped as in the previous case.
The bottleneck matching is updated according to the rules described in~\cref{lem:bot}, where we might need to test if $G'(j)$, in the notation of~\cref{lem:bot}, has a complete matching. 
If $m$ edges had weight $j$ in~$G$, the matching $\tau=\mu \cap G'(j)$ can have up to~$m$ exposed vertices. 
We search in $O(k^2)$ time for an augmenting path for $\tau$ in $G'(j)$. 
If there is one, we augment the matching and search again. 
If there is no augmenting path, the matching is complete by~\cref{lem:berge}.
Therefore, we can decide whether $G'(j)$ has a complete matching and, if not, find one performing at most as many augmentations as there are pairs in the list of~$h$.
After handling all the edge pairs of $h$, the resulting graph contains the candidate set of edges $Z(D)$ for the new cell $D$ and $w$ expresses the relative order of their lengths. 
Thus, the bottleneck matching we obtained for the last weighted graph is guaranteed to be a bottleneck matching for any point interior to $D$. 

The number of graph and matching updates performed during the traversal is bounded by the number of edges that $\lexo(A,B)$ would have if we replace each bisector associated to $s$ edge pairs by~$s$ infinitesimally-separated lines parallel to it. 
As argued before, the number of lines supporting edges of this arrangement would still be $O(nk^3)$ and, thus, the complexity of this virtual arrangement is $O(n^2k^6)$. 
Together with the required time $O(k^2)$ for updating the bottleneck matching in each cell, this implies the claimed running time $O(n^2k^8)$ for the whole algorithm.

In order to construct a lex-bottleneck labeling, we maintain the weighted graph as in the bottleneck case.
We apply the algorithm described in~\cite{Sokkalingam} in every face of $\lexo(A,B)$, after updating the weight function $w$ to indicate the ties that are active in the current face.
This computes a lexicographic bottleneck matching in $O(k^4)$ time. 
\end{proof}

%
%
%

In higher dimensions, we do not have a good bound on the complexity of the reduced arrangement $\lexo(A,B)$.
Nevertheless, the previous proof can be adapted using $\lex(A,B)$ instead, leading via \cref{cor:firstbound} to the following result.

\begin{theorem}\label{thm:constrHD}
Let $A,B \subset \R^d$ be point sets with $k=|B|\leq|A|=n$. 
A labeled bottleneck diagram of $A$ and $B$ can be computed in $O(n^{2d}k^{2d +2})$ time, and a labeled lex-bottleneck diagram can be computed in $O(n^{2d}k^{2d +4})$ time. 
\end{theorem}

\section{Applications}
\label{sect:applications}

In this section, we explore some of the applications of the bottleneck diagrams studied before. 
An obvious application is to solve the bottleneck partial-matching problem under translation. 
Furthermore, bottleneck partial-matching Voronoi diagrams serve as a data structure for dynamically querying for locally optimal bottleneck matchings.
Two such situations are described below.

\subsection{Solving the Bottleneck Partial-Matching Problem}
\label{subsect:botPMP}

Let $A, B \subset \R^2$ be two point sets with $k=|B| \leq |A|=n$.
We are interested here in finding a matching $\sigma:B\hookrightarrow A$ such that
\[ f_\sigma(t^*) = \underset{t \in \R^2} \min \; \E(t), \textrm{ for some } t^* \in \R^2, \]
using the notation from~\cref{def:bot}. 
Such a matching is called an \emph{optimal bottleneck matching under translations}. 
The basic idea to find such a matching is to traverse a labeled bottleneck diagram and compute the optimal value of the cost function in every convex cell of the diagram separately. 
Clearly, this procedure also admits to report a corresponding translation $t^*$ for which the minimum is attained. 

\begin{theorem} \label{thm:optmatch}
Let $A,B \subset \R^2$ be sets of $k$ and $n$ points. 
An optimal bottleneck matching for $A$ and $B$ under translations can be found in $O(n^2k^{8})$ time. 
\end{theorem}

\begin{proof}
We construct $\lexo(A,B)$ and a bottleneck labeling for it in time $O(n^2k^{8})$ as described in the proof of~\cref{thm:constr}.
We traverse this arrangement and optimize in every (convex) cell $C$ the value $f(t)=\|b+t-a\|^2$ over all translations $t\in C$, maintaining the minimum throughout the diagram.
Here, $ab$ is the longest edge of the bottleneck matching given by the label of the current cell~$C$.
More precisely, let $t_0=a-b$ be the translation attaining the global minimum of the function $f(t)$. 
If $t_0 \in C$, obviously $f(t_0)=0$ is the minimum of $f$ in $C$ and in fact, a global minimum as well.
Otherwise, the minimum is attained in the point of $C$ closest to $t_0$. 
Such a point must be either a vertex of $C$ or the orthogonal projection of $t_0$ onto an edge of $C$.
In addition, if $t_0$ is the minimum of $f(t)$ and $t_1$ is the minimum of the corresponding function for a neighboring cell $D$ sharing the edge $e$ with $C$, then the projection of $t_0$ onto $e$ coincides with the projection of $t_1$ onto $e$. 
We can thus calculate the minimum examining once every vertex of the diagram and at most one candidate point for every edge. 
Thus, the total time needed to perform the mentioned optimization in every cell is proportional to the complexity of the diagram. 
\end{proof}

\subsection{Computing a Bottleneck Path}
\label{subsect:botpath}

We consider now the problem of finding a motion for~$B$ from an initial position to a final position such that the maximum bottleneck value (as defined in~\cref{def:bot}) attained during the motion is minimized. 

\begin{definition}
The \emph{bottleneck value} of a curve \mbox{$\gamma: [0,1] \to \R^2$} with respect to point sets $A,B\subset\R^2$ with $k=|B| \le |A|=n$ is 
\[ F(\gamma) = \max_{s \in [0,1]} \E(\gamma(s))= \max_{s \in [0,1]} \min_{\sigma: B \hookrightarrow A} \max_{b \in B} \Vert b+\gamma(s)-\sigma(b) \Vert^2.\]
The curve $\gamma$ is a called a \emph{bottleneck path} if $F(\gamma) \leq F(\varphi)$ for every other curve $\varphi: [0,1] \to \R^2$ with $\varphi(0)=\gamma(0)$ and $\varphi(1)=\gamma(1)$. 
\end{definition}

A bottleneck path between two positions can be useful in motion planning where the points of $A$ represent fixed anchor points and the points of $B$ represent the position of articulations of a moving robot.
The dual graph of the arrangement $\lexo(A,B)$ contains the necessary information to compute a bottleneck path from any initial position to any final position. 

\begin{definition}
The \emph{bottleneck graph} of two finite point sets $A,B \subset \R^2$ is the weighted graph $\lexo(A,B)^*$ dual to $\lexo(A,B)$, where an edge $e^*$ of the graph dual to an edge $e$ of $\lexo(A,B)$ has weight $\min_{t \in e}\E(t)$.
\end{definition}

Via the bottleneck graph we can now characterize the existence of a path of given bottleneck value.

\begin{lemma}\label{lem:safpath}
Let $t_0,t_1 \in \R^2$ and $\delta \in \R$.
Let $C_0$ and $C_1$ be cells of $\lexo(A,B)$ such that $t_0\in C_0$ and $t_1 \in C_1$. 
There is a path with bottleneck value at most $\delta$ from~$t_0$ to $t_1$ if and only if $\E(t_0),\E(t_1) \leq \delta$ and there is a path from $C_0^*$ to $C_1^*$ in $\lexo(A,B)^*$ whose longest edge has weight at most $\delta$. 
\end{lemma}
\begin{proof}
Observe that in every cell of $\lexo(A,B)$ there is a bottleneck matching whose cost coincides with $\E$ in the cell. 
By definition, $\E$ is a convex function in every such (convex) cell. 
Hence, assuming that $C_0=C_1$, the line segment joining $t_0$ and $t_1$ has bottleneck value $ \max \{\E(t_0),\E(t_1)\}$ and no path can attain a smaller value. 
We assume now that $C_0 \neq C_1$ and let $\gamma$ be any path from~$t_0$ to~$t_1$.
We can replace each of the connected arcs of $\gamma$ entering a cell~$C$ of $\lexo(A,B)$ in a point $t_{\text{in}}$ and leaving it in a point~$t_{\text{out}}$ by the line segment joining these two points without increasing the bottleneck value of the path. 
Again, we do not increase the bottleneck value of the path when we substitute this line segment by the one joining the points~$t_{\text{in}}^*$ and~$t_{\text{out}}^*$, where~$t_{\text{in}}^*$ is the point with minimum bottleneck value on the edge of~$C$ that contains~$t_{\text{in}}$, and~$t_{\text{out}}^*$ is the one attaining the minimum value on the edge containing~$t_{\text{out}}$. 
Similarly, the parts of the path in~$C_0$ and~$C_1$, starting at~$t_0$ and ending at~$t_1$, respectively, can be replaced with the line segment from~$t_0$ (or~$t_1$) to the point attaining the minimum of $\E$ on whichever edge of~$C_0$ (or~$C_1$) the path crosses first (or last). 

The previous observations imply that a bottleneck path is among the poly\-gonal paths whose vertices (except for~$t_0$ and~$t_1$) lie on the minima of~$\E$ along edges of $\lexo(A,B)$. 
The bottleneck value of such a path is the maximum of the weights of the edges in the corresponding path in $\lexo(A,B)^*$ and the values~$\E(t_0)$ and~$\E(t_1)$.
\end{proof}

\begin{theorem}\label{thm:safpath}
Given $t_0,t_1 \in \R^2$, a bottleneck path from $t_0$ to~$t_1$ with respect to $A, B  \subset \R^2$ with $k=|B| \le |A|=n$ can be computed in time $O(n^2k^{6}(k^{2} +\log n))$.
\end{theorem}
\begin{proof}
We first compute the arrangement $\lexo(A,B)$ and the associated bottleneck graph $\lexo(A,B)^*$ in time $O(n^2k^8)$ by~\cref{thm:constr}. 
The number of edges and vertices of $\lexo(A,B)^*$ is $O(n^2k^6)$ due to \cref{thm:bot} and the weights of its edges are all nonnegative. 
Therefore, the path with minimum bottleneck value in the graph can be found in $O(n^2k^6 \log n)$ time via the implementation of Dijkstra's algorithm using heaps. 
By~\cref{lem:safpath}, the associated polygonal path is guaranteed to be a bottleneck path from~$t_0$ to~$t_1$.
\end{proof}

\subsection{Finding the Cover Radius of a Convex Polygon}
\label{subsect:polcovrad}

As a third application of our bottleneck diagrams, we investigate a covering problem. 
Given a pair of finite point sets $A,B \subset \R^2$, and a convex polygon $Q \subset \R^2$, we want to determine the minimal $\delta \in \R$ such that for any position of the point set $B$ in $Q$ there is a matching whose bottleneck value is at most~$\delta$.
We can think of the points in $A$ as antennas equipped with disks of radius $\delta$ modeling the region on which they provide signal. 
The point set $B$ can be thought of as a robot that moves in $Q$ and needs to connect each of its points to a different antenna (for instance, to learn its position). 
The target is to minimize the power consumed by the antennas while ensuring that the robot can move in~$Q$ having a one-to-one connection for its receivers. 

\begin{definition}
The \emph{cover radius} of a convex polygon $Q$ with respect to finite point sets $A,B \subset \R^2$ is the maximum bottleneck value among all $t\in \R^2$ such that $B+t \subset Q$. 
\end{definition}

\begin{theorem}\label{thm:coverrad}
Let $Q$ be a convex polygon with $m$ vertices. The cover radius of~$Q$ with respect to $A,B \subset \R^2$ with $k=|B|\leq |A|=n$ can be computed in time $O\!\left(n^2k^{8}+(n^2k^6+m)\log(n+m)\right)$.
\end{theorem}

\begin{proof}
We start by computing $\lexo(A,B)$ and a bottleneck labeling as indicated in \cref{thm:constr}. 
Note that the set of translations $t\in \R^2$ for which $B+t \subset Q$ is a convex polygon $\widehat{Q}$, which is the intersection of $k$ translated copies of $Q$. 
This polygon is indeed given by the $m$ linear inequalities obtained by imposing that the extreme point of $B$ in the direction orthogonal to an edge of $Q$ is on the right hand side of the corresponding edge. 
This polytope can be computed easily in $O(k \log k+m \log m)$ time and has at most $m$ edges. 
Then, we compute the overlay of the boundary of $\widehat{Q}$ and $\lexo(A,B)$. 
Note that every edge of $\lexo(A,B)$ can intersect the boundary of $\widehat{Q}$ in at most two points. 
Thus, the number of vertices of the overlay is $O(n^2k^6+m)$ and, hence, it can be computed in  $O\!\left((n^2k^6+m)\log(n+m)\right)$ time using the techniques described in \cite{Berg}.
The next step involves traversing the overlay and maintaining the maximum of~$\E$ in every cell. 
Since the function is convex, the maximum can be calculated as the maximum of the values attained at the vertices of the overlay.
\end{proof}

\section{Conclusion and Open Problems}

In this work, we introduced and investigated Voronoi-type diagrams suited for the study of the bottleneck partial-matching problem under translations.
As our main results, we obtained low complexity bounds on these diagrams and devised efficient algorithms for their construction that allowed to solve the matching problem both in arbitrary dimension and in its lexicographic variant.

We have seen that the complexity bound of $O(n^2k^6)$ in \cref{thm:bot} is sharp in terms of the parameter~$n$, i.e., the size of the bigger point set, but we do not know whether it might be improved with respect to~$k$.
Since any improvement of this kind that comes with a speed-up in the construction of the (lex-)bottleneck diagram translates into a better bound on the running time of our algorithms for the partial-matching problems, we consider this as an interesting open problem for future research.
In fact, we believe that the lower bound in \cref{lowerboundStable} is optimal.

\begin{conjecture}
For any pair of point sets $A,B \subset \R^2$ with $k=|B|\leq|A|=n$, there is a lex-bottleneck diagram of complexity $O(n^2k^2)$.
\end{conjecture}

\bibliographystyle{abbrv}
\bibliography{BT3}

\appendix

\section{\texorpdfstring{$k$}{k}-Sensitive Analysis of Previous Algorithms}
\label{appendix}

The following $k$-sensitive analysis of previous algorithms for the bottleneck partial-matching problem under translation is not intended to be a complete account of all involved details.
Rather we focus on the crucial ideas and adjustments needed in order to make these algorithms work in the general situation.

\subsection{The Algorithm of Alt, Mehlhorn, Wagener \& Welzl~\texorpdfstring{\cite{Alt}}{}}

The first algorithm was introduced in~\cite{Alt}, where the authors prove a running time of $O(n^6 \log n)$ for the balanced situation.
Their approach consists of two steps and can be adjusted toward the unbalanced case as follows.
First, for every choice of $a_1,a_2,a_3 \in A$ and $b_1,b_2,b_3 \in B$ they define $\varepsilon(b_1,a_1,b_2,a_2,b_3,a_3)$ to be the minimum $\varepsilon \in \R$ such that there is a translation placing~$b_i$ into an $\varepsilon$-neighborhood of $a_i$, for all $i=1,2,3$. 
They claim that the bottleneck distance under translations is attained by one of these $O(n^6)$ values, and they compute every such value in constant time. 
In the unbalanced case, the number of values and the time for its computation is then $O(n^3k^3)$. 
They sort these values into an array $\Ep$ and perform a binary search, testing for every $\varepsilon \in \Ep$ whether there is a bottleneck matching under translations having cost~$\varepsilon$.
In order to do that, they assume that $\|b+t-a \|^2=\varepsilon^2$, for a fixed pair $ab \in A \times B$, which restricts the set of candidate translations to a circle. 
They parametrize the circle by polar coordinates and compute the set of angles $\alpha(a',b')$ for which $b'$ lies in an $\varepsilon$-neighborhood of $a'$, for all $a'b' \in A \times B$. 
The computation of such intervals on the circle is not trivial and requires some careful observations that carry over to the unbalanced case without modification.
The arrangement in this circle induced by the circular intervals~$\alpha$ can then be computed by sorting their endpoints.
In every interval, the authors construct the bipartite graph whose edges are shorter than~$\varepsilon$ and they look for a maximum matching in it. 
This is done by computing the graph for an arbitrary initial point and traversing the circular arrangement, adding or deleting at each interval the corresponding edge or edges.
If edges are only added, nothing is left to be done.
If some edges are deleted leaving some points of $B$ unmatched, suitable augmenting-path computations need to be performed in order to decide whether there is a maximum matching in the next interval. 
The construction of each of the arrangements in the $O(nk)$ circles for a fixed value $\varepsilon \in \Ep$ can be done in $O(nk \log n)$.
The traversal requires $O(nk)$ time to construct the initial graph, $O(nk+k^{5/2})$ to prune non-relevant edges and obtain a maximum matching for the initial graph using the Hopcroft-Karp algorithm, and $O(nk)$ time per cell to compute the augmenting path to update the maximum matching.
The updates of the bipartite graph require constant time for each edge. 
Thus, the total time required for each $\varepsilon \in \Ep$ is $O(nk \cdot (nk \log n + nk+k^{5/2} +nk \cdot nk))=O(n^3k^3)$. 
Since sorting the values in~$\Ep$ is done in $O(n^3k^3 \log n)$ time, the whole algorithm runs in time $O(n^3 k^3 \log n)$. 

\subsection{The Algorithm of Efrat, Itai \& Katz~\texorpdfstring{\cite{Efrat}}{}}

An improved algorithm with a running time of $O(n^5 \log^2 n)$ in the balanced case has been described in~\cite{Efrat}.
This improvement requires the use of a non-trivial data structure and parametric search techniques combined with sorting networks. 
The data structure is used to create an oracle that, given $\varepsilon \in \Ep$ in the above notation, answers whether the bottleneck distance of~$A$ and~$B$ under translations is at most~$\varepsilon$. 
The corresponding oracle in~\cite{Alt} runs in $O(n^6)$ time and, according to the preceding analysis, can be adapted to run in $O(n^3k^3)$ in the unbalanced case. 
A gain in the computation of augmenting paths enables the authors in~\cite{Efrat} to improve this oracle to run in $O(n^5 \log n)$ time. 
This is done by utilizing a data structure based on constructing a layered graph from the matching after each possible edge deletion.
Adapting the analysis for the unbalanced case leads to an augmenting paths computation in time $O(k \log n)$ instead of the $O(nk)$ by standard techniques.
Therefore, the oracle runs in $O(n^2 k^3 \log n )$ time. 
However, the time to sort the $O(n^3 k^3)$ values into~$\Ep$ dominates the running time and thus prevents the new algorithm to improve upon the simpler one in~\cite{Alt}. 
This problem is solved by using, instead of the binary search, an adaptation of the parametric search technique due to Cole in order to avoid the construction of~$\Ep$, and hence reducing the number of oracle calls to $O(\log n)$. 
Therefore, the final running time for the general unbalanced situation is $O(n^2 k^3 \log^2 n)$.

\end{document}